\documentclass[reqno,10pt]{amsart}

\usepackage{amsmath,amssymb,amsthm,mathrsfs,enumitem}
\usepackage{upgreek}
\usepackage{slashed}
\usepackage{xifthen}
\usepackage{graphicx}
\usepackage{longtable}

\usepackage{xcolor}
\definecolor{winered}{rgb}{0.7,0,0}
\definecolor{lessblue}{rgb}{0,0,0.7}

\usepackage[pdftex,colorlinks=true,linkcolor=winered,citecolor=lessblue,breaklinks=true,bookmarksopen=true]{hyperref}

\usepackage{titletoc}

\makeatletter

\def\@tocline#1#2#3#4#5#6#7{
\begingroup
  \par
    \parindent\z@ \leftskip#3 \relax \advance\leftskip\@tempdima\relax
                  \rightskip\@pnumwidth plus 4em \parfillskip-\@pnumwidth
    \ifcase #1 
       \vskip 0.6em \hskip 0em 
       \or
       \or \hskip 0em 
       \or \hskip 1em 
    \fi%
    #6
    \nobreak\relax{\leavevmode\leaders\hbox{\,.}\hfill}
    \hbox to\@pnumwidth {\@tocpagenum{#7}}
  \par
\endgroup
}

 \def\l@section{\@tocline{0}{0pt}{0pc}{}{}}

\renewcommand{\tocsection}[3]{%
  \indentlabel{\@ifnotempty{#2}{ 
    \ignorespaces\bfseries{#2. #3}}}
  \indentlabel{\@ifempty{#2}{\ignorespaces\bfseries{#3}}{}} 
    \vspace{1.5pt}}

\renewcommand{\tocsubsection}[3]{%
  \indentlabel{\@ifnotempty{#2}{
    \ignorespaces#2. #3}}
  \indentlabel{\@ifempty{#2}{\ignorespaces #3}{}}
    \vspace{1.5pt}}

\renewcommand{\tocsubsubsection}[3]{%
  \indentlabel{\@ifnotempty{#2}{
    \ignorespaces#2. #3}}
  \indentlabel{\@ifempty{#2}{\ignorespaces #3}{}}
    \vspace{1.5pt}}

\makeatother

\makeatletter
\def\@nomenstarted{0}
\newlength{\@nomenoldtabcolsep}

\newcommand{\nomenstart}
  {%
    \def\@nomenstarted{1}%
    \setlength{\@nomenoldtabcolsep}{\tabcolsep}%
    \setlength{\tabcolsep}{3.5pt}%
    \begin{longtable}{p{0.11\textwidth} p{0.86\textwidth}}
  }

\newcommand{\nomenitem}[2]{%
    \ifcase\@nomenstarted%
      \or 
      \or \\ 
    \fi%
    #1\,{\leavevmode\leaders\hbox{\,.}\hfill} & #2%
    \def\@nomenstarted{2}%
  }%
\newcommand{\nomenend}
  {\\%
      \end{longtable}%
      \setlength{\tabcolsep}{\@nomenoldtabcolsep}%
      \def\@nomenstarted{0}%
  }
\makeatother

\allowdisplaybreaks

\numberwithin{equation}{section}
\newtheorem{thm}{Theorem}[section]
\newtheorem{prop}[thm]{Proposition}
\newtheorem{lemma}[thm]{Lemma}

\newtheorem*{thm*}{Theorem}
\newtheorem*{prop*}{Proposition}
\newtheorem*{cor*}{Corollary}
\newtheorem*{conj*}{Conjecture}

\theoremstyle{definition}

\theoremstyle{remark}

\newcommand{\mc}{\mathcal}

\newcommand{\cC}{\mc C}

\newcommand{\cH}{\mc H}

\newcommand{\cL}{\mc L}

\newcommand{\cU}{\mc U}
\newcommand{\cV}{\mc V}
\newcommand{\cW}{\mc W}

\newcommand{\ms}{\mathscr}

\newcommand{\sR}{\ms R}

\newcommand{\C}{\mathbb{C}}
\newcommand{\N}{\mathbb{N}}
\newcommand{\R}{\mathbb{R}}
\newcommand{\Z}{\mathbb{Z}}
\newcommand{\Sph}{\mathbb{S}}

\newcommand{\bfzero}{\mathbf{0}}
\newcommand{\bfa}{\mathbf{a}}

\newcommand{\bfB}{\mathbf{B}}
\newcommand{\bfE}{\mathbf{E}}

\renewcommand{\Im}{\operatorname{Im}}
\newcommand{\Id}{\operatorname{Id}}

\newcommand{\tr}{\operatorname{tr}}

\newcommand{\eps}{\epsilon}

\newcommand{\hra}{\hookrightarrow}

\newcommand{\ol}{\overline}
\newcommand{\pa}{\partial}

\newcommand{\wt}{\widetilde}

\newcommand{\bop}{{\mathrm{b}}}

\newcommand{\rcTbdual}[1][]{\ensuremath{\overline{{}^{\bop}T^*\ifthenelse{\isempty{#1}}{}{_#1}}}}

\newcommand{\bhm}[1][]{\ensuremath{M_{\bullet\ifthenelse{\isempty{#1}}{}{,#1}}}}

\newcommand{\CI}{\cC^\infty}

\newcommand{\cCb}{\mc C_\bop}
\newcommand{\CIb}{\cC^\infty_\bop}

\newcommand{\Ric}{\mathrm{Ric}}

\newcommand{\openbigpmatrix}[1]
  {%
    \def\@bigpmatrixsize{#1}%
    \addtolength{\arraycolsep}{-#1}%
    \begin{pmatrix}%
  }
\newcommand{\closebigpmatrix}
  {%
    \end{pmatrix}%
    \addtolength{\arraycolsep}{\@bigpmatrixsize}%
  }

\newcommand{\itref}[1]{(\ref{#1})}

\DeclareGraphicsExtensions{.mps}
\newcommand{\inclfig}[1]{\includegraphics{#1}}

\begin{document}

\title[Uniqueness of Kerr--Newman--de~Sitter]{Uniqueness of Kerr--Newman--de~Sitter black holes with small angular momenta}

\author{Peter Hintz}
\address{Department of Mathematics, University of California, Berkeley, CA 94720-3840, USA}
\email{phintz@berkeley.edu}

\date{February 16, 2017. Final revision: November 24, 2017.}


\begin{abstract}
  We show that a stationary solution of the Einstein--Maxwell equations which is close to a non-degenerate Reissner--Nordstr\"om--de~Sitter solution is in fact equal to a slowly rotating Kerr--Newman--de~Sitter solution. The proof uses the non-linear stability of the Kerr--Newman--de~Sitter family of black holes with small angular momenta, recently established by the author, together with an extension argument for Killing vector fields. Our black hole uniqueness result only requires the solution to have high but finite regularity; in particular, we do not make any analyticity assumptions.
\end{abstract}

\maketitle

\section{Introduction}
\label{SecIntro}

Let $(M,g,F)$ be a smooth stationary solution of the Einstein--Maxwell system
\begin{equation}
\label{EqIntroEM}
  \Ric(g) + \Lambda g = 2 T(g,F), \quad d F=0, \quad \delta_g F=0,
\end{equation}
with $M$ a 4-manifold, $g$ a Lorentzian metric with signature $(+,-,-,-)$, and $F$ a 2-form on $M$; here $T(g,F)_{\mu\nu}=-F_{\mu\lambda}F_\nu{}^\lambda+\frac{1}{4}g_{\mu\nu}F_{\kappa\lambda}F^{\kappa\lambda}$ is the electromagnetic energy-momentum tensor, and $\Lambda$ is the cosmological constant, which we will take to be \emph{positive} and fixed. If $(M,g,F)$ is close to a non-degenerate Reissner--Nordstr\"om--de~Sitter (RNdS) black hole in a neighborhood of the domain of outer communications (more precisely, close to the future development of RNdS initial data on a spacelike hypersurface intersecting and extending a bit past the future event and cosmological horizons), we show that $(M,g,F)$ is isometric to a slowly rotating Kerr--Newman--de~Sitter (KNdS) black hole. We deduce this from the full non-linear stability result proved by the author in \cite{HintzKNdSStability}. To state our theorem, we introduce
\begin{gather*}
  M_\eps = [0,\infty)_{t_*} \times (r_0+\eps,r_1-\eps)_r \times \Sph^2, \\
  \Sigma_\eps = \{0\} \times (r_0+\eps,r_1-\eps) \times \Sph^2 \subset M_\eps,
\end{gather*}
with $0<r_0<r_1$ specified below, and
\[
  M = M_0,\quad \Sigma=\Sigma_0.
\]
Denote the parameters of a KNdS black hole by $b=(\bhm,\bfa,Q_e,Q_m)$, where $\bhm>0$ is its mass, $\bfa\in\R^3$ its angular momentum per unit mass, and $Q_e,Q_m\in\R$ are its electric and magnetic charge. Let $b_0=(\bhm[0],\bfzero,Q_{e,0},Q_{m,0})$ denote the parameters of a fixed non-degenerate RNdS black hole, see \cite[Definition~3.1]{HintzKNdSStability}, which thus has an event horizon at $r=r_{b_0,-}$ and a cosmological horizon at $r=r_{b_0,+}>r_{b_0,-}$. We choose $r_0$ and $r_1$ such that
\[
  r_0 + 8\rho < r_{b_0,-} < r_{b_0,+} < r_1 - 8\rho
\]
for a fixed small $\rho>0$, and such that moreover the Cauchy horizon (which exists if $Q_e^2+Q_m^2>0$) of the RNdS black hole lies in $r<r_0-\rho$. We only consider parameters $b$ which are close to $b_0$. The KNdS family, recalled below, is then a smooth family $(g_b,F_b)$ of solutions of \eqref{EqIntroEM} on the fixed manifold $M$; the vector field $\pa_{t_*}$ is Killing in the sense that $\cL_{\pa_{t_*}}g_b=0$ and $\cL_{\pa_{t_*}}F_b=0$ for all $b$. We assume that the event, resp.\ cosmological, horizon of $g_b$ is located at $r=r_{b,-}$, resp.\ $r=r_{b,+}$, with $r_{b,\pm}\in(r_0+8\rho,r_1-8\rho)$, and that the Cauchy horizon of $g_b$ is located in $r<r_0-\rho$; this holds for $b$ sufficiently close to $b_0$.

Lastly, if $\psi\colon\Sigma_\eps\to M$ is any smooth map, so $\psi(x)=(\psi^T(x),\psi^X(x))$ with $\psi^T$ real-valued and $\psi^X$ taking values in $(r_0,r_1)\times\Sph^2$, we define its \emph{stationary extension} to be the map $\wt\psi\colon[0,\infty)_{t_*}\times\Sigma_\eps \to M$,
\[
  \wt\psi(t_*,x) = (\psi^T(x)+t_*,\psi^X(x)),\quad x\in\Sigma_\eps.
\]

\begin{thm}
\label{ThmIntroUniq}
  Suppose $(g,F)$ is a smooth solution of the Einstein--Maxwell system \eqref{EqIntroEM} on $M$ which is stationary with respect to $\pa_{t_*}$, i.e.\ $\cL_{\pa_{t_*}}(g,F)=0$. Suppose moreover that $(g,F)$ is close to $(g_{b_0},F_{b_0})$ in the topology of $\cC^m$ for some large fixed $m$. Then there exist parameters $b$ close to $b_0$ and a smooth hypersurface $\Sigma'\subset M$, together with a smooth map $\psi\colon\Sigma_{7\rho}\to\Sigma'\subset M$ $\rho$-close in $\cC^0$ to the map $(t_*,x)\mapsto(t_*+1,x)$ such that for the stationary extension $\wt\psi$ of $\psi$, we have
  \[
    \wt\psi^*(g_b,F_b) = (g,F).
  \]
\end{thm}

See Figure~\ref{FigIntroUniq}. Here, we define $\cC^m$ using the embedding $M\hra\R^4$ by means of polar coordinates. If one merely assumes $(g,F)$ to have fixed but high regularity, the conclusion remains valid, only with $\psi$ having finite (but high) regularity only. The conclusion cannot be strengthened: If $\psi\colon\Sigma_\eps\to M$, $\eps\geq 0$ small, is any smooth map, close to the time $1$ translation map, with stationary extension $\wt\psi$, then $(g,F)=\wt\psi^*(g_b,F_b)$ satisfies the assumptions of the theorem (after slightly shrinking the size of the interval $[r_0,r_1]$ and thus of $M$), and the conclusion holds for $\psi$ and $\Sigma'=\psi(\Sigma)$.

\begin{figure}[!ht]
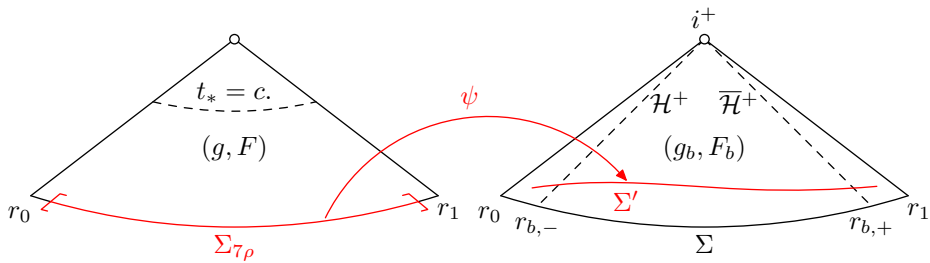

\centering
\inclfig{IntroUniq-small}
\caption{\textit{Left:} The given stationary solution $(g,F)$, drawn in a Penrose diagrammatic fashion together with a level set of $t_*$. \textit{Right:} Penrose diagram of a KNdS solution $(g_b,F_b)$, with event horizon $\cH^+$, cosmological horizon $\ol\cH{}^+$, and future infinity $i^+$. \textit{Red:} The stationary extension $\wt\psi$ of the map $\psi\colon\Sigma_{7\rho}\to\Sigma'$ pulls the KNdS solution back to the given one.}
\label{FigIntroUniq}
\end{figure}

In Boyer--Lindquist coordinates, a KNdS solution $(g_b,F_b)$ with magnetic charge $Q_m=0$ and $a=|\bfa|$ takes the form
\begin{align*}
  g_b &= -\rho_b^2\Bigl(\frac{dr^2}{\wt\mu_b}+\frac{d\theta^2}{\kappa_b}\Bigr) + \frac{\wt\mu_b}{(1+\lambda_b)^2\rho_b^2}(dt-a\sin^2\theta\,d\phi)^2 \\
    &\qquad - \frac{\kappa_b\sin^2\theta}{(1+\lambda_b)^2\rho_b^2}(a\,dt-(r^2+a^2)\,d\phi)^2,
\end{align*}
where we use polar coordinates $\phi\in\R/2\pi\Z$ and $\theta\in(0,\pi)$ on $\Sph^2$, and
\begin{gather*}
  \lambda_b=\frac{\Lambda a^2}{3},\quad \kappa_b=1+\lambda\cos^2\theta, \quad \rho_b^2=r^2+a^2\cos^2\theta, \\
  \wt\mu_b=(r^2+a^2)\Bigl(1-\frac{\Lambda r^2}{3}\Bigr)-2\bhm r+(1+\lambda_b)^2 Q_e^2;
\end{gather*}
the radii $r_{b,-}<r_{b,+}$ are the two largest roots of $\wt\mu_b$. The 4-potential is
\[
  \breve A_b = -\frac{Q_e r}{\rho_b^2}(dt-a\sin^2\theta\,d\phi),
\]
and the electromagnetic tensor is correspondingly given by $F_b=d\breve A_b$. These expressions for $g_b$ and $F_b$ are singular at the horizons $r_{b,\pm}$, but a simple change of coordinates removes these singularities, giving an extension of $(g_b,F_b)$ to $M$; see \cite[\S3.2]{HintzKNdSStability} for details. KNdS solutions with non-zero magnetic charge can be constructed from magnetically uncharged solutions using the electric-magnetic duality, as described in \cite[\S3.3]{HintzKNdSStability}.

Theorem~\ref{ThmIntroUniq} is a rigidity theorem for stationary black holes in the Einstein--Maxwell theory in the \emph{smooth} category and \emph{without symmetry assumptions}. In the case of $\Lambda=0$, rigidity of the Kerr--Newman family was proved by Wong in his thesis \cite{WongThesisKerrNewman} --- conditional on a technical assumption, see \cite[item (T) in \S4.4]{WongThesisKerrNewman} --- using a tensorial characterization of Kerr--Newman black holes \cite{WongKerrNewmanChar}. These results built on the work on rigidity theorems for Kerr spacetimes, in the context of the Einstein vacuum equations and under smallness assumptions on the angular momentum or the Mars--Simons tensor, by Alexakis--Ionescu--Klainerman \cite{AlexakisIonescuKlainermanKerrUniquenessSmallMarsSimons,AlexakisIonescuKlainermanUniqueness}, which in turn relied on earlier work by Ionescu--Klainerman \cite{IonescuKlainermanUniqueness}. See also \cite{WongYuKerrNewmanNoMult}. These works make certain geometric assumptions in particular on the structure of the future and past event horizons. While we do not explicitly make such assumptions in the statement of Theorem~\ref{ThmIntroUniq}, they are in fact automatically satisfied, as a careful analysis of the solution produced by the main theorem of \cite{HintzKNdSStability} shows; this will be discussed elsewhere.

Earlier uniqueness results for the Kerr spacetime by Hawking and Ellis \cite{HawkingEllis}, Carter \cite{CarterAxisymmetry}, and Robinson \cite{RobinsonKerrUniqueness}, as well as later work by Mazur \cite{MazurKerrNewmanUniqueness} on the rigidity of Kerr--Newman black holes, relied on the restrictive assumption of real analyticity of the metric and the electromagnetic field. We refer the reader to the reviews \cite{RobinsonBlackHoleUniquenessReview,ChruscielCostaHeuslerStationaryBH} for further results and references.

Proving uniqueness for the \emph{linearized} Einstein(--Maxwell) equations amounts to performing a mode stability analysis at $0$ frequency, which was done by Vishveshwara \cite{VishveshwaraSchwarzschild}, Moncrief \cite{MoncriefRNGaugeInv}, and Kodama \cite{KodamaUniqueness} in various contexts. Via the non-linear stability theorem, a key ingredient of our proof of Theorem~\ref{ThmIntroUniq} is the analysis of (generalized) $0$ frequency modes for linear perturbations of non-degenerate RNdS black holes in the spirit of the work by Kodama--Ishibashi \cite{KodamaIshibashiMaster,KodamaIshibashiCharged} on perturbations of Schwarzschild--de~Sitter and RNdS spacetimes; we do however need the full strength of \cite{HintzKNdSStability} in the argument presented here.

Using the Kerr--de~Sitter (KdS) stability result with Vasy \cite{HintzVasyKdSStability}, our arguments also prove the uniqueness of slowly rotating KdS black holes. Furthermore, as pointed out in \cite[Remark~1.2]{HintzKNdSStability}, non-linear stability for perturbations of slowly rotating KdS or KNdS black holes also holds for the Einstein(--Maxwell) equations coupled to a massless or massive scalar field. Correspondingly, perturbative uniqueness results, analogous to Theorem~\ref{ThmIntroUniq}, hold for such black holes. (When coupling to a massive scalar field, the restriction to small angular momenta in Theorem~\ref{ThmIntroUniq} is likely necessary --- with the smallness depending on the mass of the scalar field ---, as the construction of time-periodic perturbations of Kerr black holes by Chodosh and Shlapentokh-Rothman \cite{ChodoshShlapentokhRothmanPeriodic} suggests.)

We close this introduction by recalling the non-linear stability result for KNdS black holes \cite{HintzKNdSStability}: We take $\Sigma=\{t_*=0\}$ as the Cauchy surface for the formulation of the initial value problem of the Einstein--Maxwell system. We denote initial data on $\Sigma$ by $(h,k,\bfE,\bfB)$, where $h$ is a Riemannian metric on $\Sigma$, $k$ a symmetric 2-tensor, and $\bfE$ and $\bfB$ are 1-forms, assumed to be smooth. A metric $g$ and a 2-form $F$ on $M$, satisfying the Einstein--Maxwell system \eqref{EqIntroEM}, are said to induce these data provided $h=-i^*g$, with $i\colon\Sigma\hra M$ the inclusion, and $k$ are the metric and second fundamental form of $\Sigma$ induced by $g$, respectively, and if for the future unit normal $N$ to $\Sigma$, we have $\bfE=-i^*i_N F$ and $\bfB=\star_h i^* F$. The \emph{constraint equations} are necessary and sufficient conditions for a local solution of this initial value problem to exist; see \cite[equations~(2.14)--(2.15)]{HintzKNdSStability} for their form in the present context.

\begin{thm}
\label{ThmIntroStab}
  (See \cite[Theorem~9.2 and Remark~9.3]{HintzKNdSStability}.) Fix the parameters $b_0$ of a non-degenerate RNdS spacetime. Then there exists $\alpha>0$ such that the following holds: Suppose $(h,k,\bfE,\bfB)$ are smooth initial data on $\Sigma$, satisfying the constraint equations, which are close, in the topology of $H^{21}$, to the initial data induced by the non-degenerate RNdS solution $(g_{b_0},F_{b_0})$. Then there exist a smooth global forward solution $(g,F)$ of the Einstein--Maxwell system, attaining the given initial data on $\Sigma$, and KNdS parameters $b=(\bhm,\bfa,Q_e,Q_m)$ close to $b_0$, such that
  \[
    g = g_b + \wt g, \quad A = A_b + \wt A,
  \]
  where $\wt g\in e^{-\alpha t_*}\CIb(M;S^2 T^*M)$ and $\wt F\in e^{-\alpha t_*}\CIb(M;\Lambda^2 T^*M)$ are exponentially decaying. Moreover, $b$, $\wt g$ and $\wt F$ depend continuously on $(h,k,\bfE,\bfB)$.
\end{thm}

Here, $\CIb$ denotes bounded functions (or sections of a vector bundle) which are bounded together will all their coordinate derivatives, using the embedding $M\hra\R^4$. The continuous dependence is not explicitly stated in the reference, but is automatic for the Nash--Moser iteration scheme employed there, see also \cite[Theorem~11.2]{HintzVasyKdSStability}. Since we do not track the number of derivatives on the data which need to be controlled in order to control a fixed $\cC^m$ norm of $(\wt g,\wt A)$, we work with $\cC^m$ spaces of high but unspecified regularity $m\in\N$ in Theorem~\ref{ThmIntroUniq} and in our arguments below; in principle however one could specify a numerical value for $m$ by a careful study of the proof of the Nash--Moser iteration scheme.

In order to prove Theorem~\ref{ThmIntroUniq}, we use this non-linear stability result as follows: If in Theorem~\ref{ThmIntroStab} we take as data the initial data of a stationary solution $(g,F)$ which are close to non-degenerate RNdS data, then the solution $(g',F')$ given by Theorem~\ref{ThmIntroStab} asymptotes to a stationary slowly rotating KNdS solution $(g_b,F_b)$; however, $(g',F')$ itself is stationary --- since it is the same as $(g,F)$, up to diffeomorphism equivalence --- and hence must in fact be equal to $(g_b,F_b)$, again up to diffeomorphism equivalence. We implement this idea in \S\S\ref{SecExt} and \ref{SecPf}.

\subsection*{Acknowledgments}

I am grateful to Andr\'as Vasy and Maciej Zworski for useful discussions, and to the Miller Institute at the University of California, Berkeley, for support.

\section{Extension of Killing vector fields}
\label{SecExt}

Given a solution $(g,F)$ of equation~\eqref{EqIntroEM} on a globally hyperbolic spacetime $(M,g)$ with Cauchy surface $\Sigma\subset M$, we show:

\begin{prop}
\label{PropExt}
  Let $\cU$ be an open neighborhood of $\Sigma$, and suppose that $V_0\in\cV(\cU)$ is a smooth vector field such that $\cL_{V_0} g=0$ and $\cL_{V_0} F=0$ in $\cU$. Then there exists a unique global extension $V\in\cV(M)$ of $V_0$ such that $\cL_V g=0$ and $\cL_V F=0$.
\end{prop}
\begin{proof}
  Let $V$ be any vector field, and define the 2-tensor, resp.\ 2-form
  \[
    \pi := \cL_V g, \quad \varphi := \cL_V F.
  \]
  Denote by $\phi_s$ the time $s$ flow of $V$. Due to the diffeomorphism invariance of the Einstein--Maxwell system, we have
  \begin{align*}
    0 &= \frac{d}{ds}\bigl(\Ric(\phi_s^*g)+\Lambda\phi_s^*g-2T(\phi_s^*g,\phi_s^*F)\bigr)\big|_{s=0} \\
      &= D_g\Ric(\pi) + \Lambda\pi - 2 D_g T(\pi,F) - 2 D_F T(g,\varphi).
  \end{align*}
  We then recall from \cite[equation~(2.4)]{GrahamLeeConformalEinstein} the formula $D_g\Ric=\frac{1}{2}\Box_g - \delta_g^*\delta_g G_g + \sR_g$, where $G_g u=u-\frac{1}{2}g\tr_g u$ is the trace reversal, $(\delta_g^*w)_{\mu\nu}=\frac{1}{2}(w_{\mu;\nu}+w_{\nu;\mu})$ the symmetric gradient, and
  \[
    (\sR_g h)_{\mu\nu} = R_{\mu\kappa\lambda\nu}h^{\kappa\lambda} + \frac{1}{2}(\Ric_{\mu\kappa}h_\nu{}^\kappa + \Ric_{\nu\kappa}h_\mu{}^\kappa)
  \]
  is a zeroth order operator. Aiming to derive a wave equation for $\pi$, we note that
  \[
    \delta_g G_g\pi = \delta_g G_g\cL_V g=2\delta_g G_g\delta_g^* V = (\Box_g - \Ric(g))V.
  \]
  Thus, if $V$ satisfies the equation
  \begin{equation}
  \label{EqExtVWave}
    (\Box_g - \Ric(g))V = 0,
  \end{equation}
  then $\delta_g G_g\pi\equiv 0$, hence $\pi$ satisfies the wave equation
  \begin{equation}
  \label{EqExtPiWave}
    \bigl(\Box_g + 2\Lambda + 2\sR_g - 4 D_g T(\cdot,F)\bigr)(\pi) = 4 D_F T(g,\varphi).
  \end{equation}
  (Conversely, if $V$ is a Killing extension of $V_0$, then $\pi=0$, hence $V$ satisfies equation~\eqref{EqExtVWave}, which proves the uniqueness of the extension $V$.)

  Next, we observe that $d\phi_s^*F\equiv 0$, hence $d\varphi=0$, and
  \[
    0 = \frac{d}{ds}(\delta_{\phi_s^*g}\phi_s^*F)|_{s=0} = \delta_g\varphi + (D_g\delta_{(\cdot)}(\pi))(F),
  \]
  hence, we have the wave equation
  \begin{equation}
  \label{EqExtPhiWave}
    (d\delta_g+\delta_g d)\varphi = -d(D_g\delta_{(\cdot)}(\pi))(F),
  \end{equation}
  for any choice of $V$. Note that while up to second derivatives of $\pi$ appear in this equation, $\varphi$ in \eqref{EqExtPiWave} appears only undifferentiated. Thus, energy estimates do apply to the coupled system \eqref{EqExtPiWave}--\eqref{EqExtPhiWave} of wave equations: One estimates $\pi$ in $H^s$ and $\varphi$ in $H^{s-1}$. (See the discussion around \cite[equation~(8.42)]{TaylorPDE}.)

  To prove the existence of $V$, we solve the wave equation \eqref{EqExtVWave} with $V=V_0$ near $\Sigma$; since we have $\pi=0$ and $\varphi=0$ near $\Sigma$, and since $\pi$ and $\varphi$ solve the coupled system \eqref{EqExtPiWave}--\eqref{EqExtPhiWave} of linear wave equations, we conclude that $\pi\equiv 0$ and $\varphi\equiv 0$ on $M$.
\end{proof}

\section{Proof of the main theorem}
\label{SecPf}

Let $D:=(h,k,\bfE,\bfB)$ denote the initial data induced on $\Sigma$ by the given smooth stationary solution $(g,F)$. By assumption, $D$ is close (in a high regularity $\cC^m$ norm) to the data $D_{b_0}$ induced on $\Sigma$ by the RNdS solution $(g_{b_0},F_{b_0})$. Let $(g',F')$ denote the forward solution of the Einstein--Maxwell system, with initial data $D$ on $\Sigma$, given by Theorem~\ref{ThmIntroStab}, thus there exist black hole parameters $b$ with
\begin{equation}
\label{EqPfExpDec}
  \wt g:=g'-g_b \in e^{-\alpha t_*}\CIb(M;S^2 T^*M),\ \ \wt F:=F'-F_b \in e^{-\alpha t_*}\CIb(M;T^*M).
\end{equation}
By uniqueness of solutions to the initial value problem, $(g',F')$ and $(g,F)$ represent the same solution. More precisely:
\begin{lemma}
\label{LemmaPfSame}
  Fix $m\in\N$, then there exists $m'\in\N$ such that the following holds: If $D$ is sufficiently close to $D_{b_0}$ in $\cC^{m'}$, then there exist a neighborhood $\cU$ of $\Sigma_{4\rho}$ in $M$ and a diffeomorphism $\phi\colon\cU\to\phi(\cU)$ with $\phi|_{\Sigma_{4\rho}}=\Id$ and such that $\phi^*g'=g$ and $\phi^*F'=F$. Embedding $M\hra\R^4$ by means of $(t_*,r,\omega)\mapsto(t_*,r\omega)\in\R\times\R^3$, the smooth function $\phi-\Id$ is small in $\cC^m(\cU;\R^4)$.

  Moreover, the Killing vector field $V=\phi_*(\pa_{t_*})$ for $(g',F')$, defined in $\phi(\cU)$, takes the form $V=(1+f)\pa_{t_*}+V_1$, with the smooth function $f$ and the smooth vector field $V_1$ small in $\cC^m(\phi(\cU))$.
\end{lemma}
\begin{proof}
  Denote by $N$ and $N'$ the future unit normal vector fields on $\Sigma$ with respect to $g$ and $g'$, respectively. For small $\delta>0$, the normal exponential map $\Sigma_\rho\times(-\delta,\delta)\ni(x,s)\mapsto\exp^g_x(s N)$ (defined with respect to the metric $g$) is a diffeomorphism onto its image; pulling $g$ back by this map, we have
  \[
    g = \begin{pmatrix} 1 & 0 \\ 0 & -h \end{pmatrix},
  \]
  and at $\Sigma$
  \[
    k = -\frac{1}{2}\pa_s h|_{s=0}, \quad F = \bfE \wedge ds - \star_h \bfB.
  \]
  We can perform the same construction for the metric $g'$ and the map $(x',s')\mapsto\exp_{x'}^{g'}(s' N')$. Define the diffeomorphism $\phi_0\colon(x,s)\mapsto(x',s')$ by $x'=x$, $s'=s$; this is well-defined in a neighborhood of $\Sigma_{2\rho}$. Since $g$ and $g'$ induce the same initial data pointwise on $\Sigma$, $g$ and $g'_0:=\phi_0^*g'$, defined near $\Sigma_{2\rho}$, have the same 1-jet on $\Sigma_{2\rho}$. By the same token, $F$ and $F'_0:=\phi_0^*F'$ agree on $\Sigma_{2\rho}$; in fact, since they satisfy Maxwell's equations with respect to the metrics $g$ and $g'_0$, respectively, their 1-jets agree.
  
  Define now the map $\phi^1$ to be the solution of the (semilinear) wave map equation
  \[
    \Box_{g'_0,g}\phi^1 = 0,
  \]
  with trivial initial data, i.e.\ $\phi^1|_{\Sigma_{3\rho}}=\Id$ and $D\phi^1|_{\Sigma_{3\rho}}=\Id_{T\Sigma_{3\rho}}$. (See \cite[Remark~2.1]{HintzVasyKdSStability} for the notation.) On a neighborhood of $\Sigma_{3\rho}$, $\phi^1$ is a diffeomorphism onto its image, which we denote by $\cU$. Let $\phi_1=(\phi^1)^{-1}$ and $\phi=\phi_0\circ\phi_1$, then the identity map $\Id_\cU\colon\cU\to\cU$ induces a wave map $(\cU,\phi^*g')=(\cU,\phi_1^*g'_0)\to(\cU,g)$; of course, $\Id_\cU$ also induces a wave map $(\cU,g)\to(\cU,g)$. By the local (pointwise!) uniqueness for the  Einstein--Maxwell system in this wave map gauge, i.e.\ with $g$ as a background metric (see \cite[\S2.2]{HintzKNdSStability}), we conclude that $(g,F)=\phi^*(g',F')$ near $\Sigma_{4\rho}$, as desired.

  The smallness claims follow from the continuity part of Theorem~\ref{ThmIntroStab}, together with the observation that the above construction in the case $(g',F')=(g,F)$ produces the map $\phi\equiv\Id$.
\end{proof}

Since the Killing vector field $V$ on $(g',F')$ is in general \emph{not} equal to $\pa_{t_*}$, the global existence of $V$ as a Killing field is not automatic at this point; but using Proposition~\ref{PropExt} for $(g',F')$, i.e.\ solving the wave equation
\begin{equation}
\label{EqPfExtVWave}
  (\Box_{g'}-\Ric(g'))V = 0,
\end{equation}
we can extend $V$ from a neighborhood of $\Sigma_{4\rho}$ to a Killing vector field satisfying $\cL_V F'=0$. Here, we assume that the future domain of dependence of $\Sigma_{4\rho}$ with respect to $g'$ contains a neighborhood $M_{5\rho}$ of the black hole exterior region, which holds for sufficiently small perturbations of the RNdS data $D_{b_0}$; the vector field $V$ is then defined in $M_{5\rho}$.

In what follows, the exponential decay rate $\alpha>0$ will be made smaller as necessary, hence may change from line to line.

\begin{lemma}
\label{LemmaPfKillingAsymp}
  The vector field $V$ takes the form
  \begin{equation}
  \label{EqPfKilling}
    V = C\pa_{t_*} + R + \wt V,
  \end{equation}
  where $C>0$ is a constant, $R\in\cV(\Sph^2)\subset\cV(M_{5\rho})$ is a rotational Killing field, and $\wt V\in e^{-\alpha t_*}\CIb(M_{5\rho};T M_{5\rho})$. Moreover, $|C-1|$ is small, $R$ is a small rotation, and $\wt V$ has small norm in $e^{-\alpha t_*}\cC^0(M_{5\rho};T M_{5\rho})$.
\end{lemma}
\begin{proof}
  Assume first that $g'$ is stationary, or more generally a smooth b-metric on the compactification of $M_{5\rho}$ defined by letting $\tau=e^{-t_*}$, and adding $\tau=0$ as a boundary at future infinity; equivalently, assume that the coefficients of $g'$ (and thus of $\wt g$ in \eqref{EqPfExpDec}), as a metric on $\R^4$, extend to smooth functions of $(\tau,r,\omega)$ down to $\tau=0$. Since equation~\eqref{EqPfExtVWave} agrees with the tensor wave operator on 1-forms up to sub-sub\-prin\-ci\-pal terms, the main result of \cite{HintzPsdoInner} applies, showing that $V$ has an asymptotic expansion
  \begin{equation}
  \label{EqPfResExp}
    V(t_*,x) = \sum_{j=1}^N V_j(t_*,x) + \wt V(t,x), \ \ V_j(t_*,x)=\sum_{l=0}^{d_j} e^{-i\sigma_j t_*}t_*^l V_{j l}(x),
  \end{equation}
  where $N\in\Z_{\geq 0}$, the $\sigma_j\in\C$ are pairwise distinct resonances, $\Im\sigma_j\geq 0$, $V_{jk}\in\CI(\Sigma_{5\rho};T_{\Sigma_{5\rho}}M_{5\rho})$, $V_{j d_j}\not\equiv 0$, and $\wt V\in e^{-\alpha t_*}\CIb(M_{5\rho};T M_{5\rho})$; here $\alpha>0$ is a fixed small number. Since $g'$ is stationary up to a remainder which is exponentially decaying in $t_*$, the Killing equation $\cL_V g'=2\delta_{g'}^* V=0$ implies that for each $j=1,\ldots,N$, the generalized mode $V_j$ satisfies the Killing equation $\delta_{g_b}^*V_j=0$. However, the only Killing vector fields of $g_b$ are $\pa_{t_*}$ and $\pa_\phi$ if the black hole with parameters $b$ has non-zero angular momentum, and all rotational vector fields on $\Sph^2$ for non-rotating black holes. This forces $\sigma_j=0$ and $d_j=1$, hence we conclude that \eqref{EqPfKilling} holds for some $C\in\R$.

  Observe now that if the initial data were the ones induced by the RNdS solution $(g_{b_0},F_{b_0})$, then we would have $(g',F')=(g_{b_0},F_{b_0})$, the local diffeomorphism $\phi$ constructed in Lemma~\ref{LemmaPfSame} would be the identity, and the Killing vector field $V$ would be $V=\pa_{t_*}$, so $C=1$, $R=0$, and $\wt V=0$ in \eqref{EqPfKilling}. The smallness assertions then follow from the continuous dependence of the resonance expansion \eqref{EqPfResExp} on the coefficients of the operator in equation~\eqref{EqPfExtVWave} and on the initial data, see \cite[\S5.1.2]{HintzVasyKdSStability} for details.

  In general, we only have the regularity stated in~\eqref{EqPfExpDec} for $g'$, i.e.\ conormality instead of smoothness on the compactification of $M_{5\rho}$. In this case, one first notes that energy estimates immediately give $V\in e^{\ell t_*}L^2(M_{5\rho};T M_{5\rho})$ for some $\ell\gg 0$, and in fact $V\in e^{\ell t_*}\CIb$ by propagation of b-singularities, see \cite[Lemma~5.2]{HintzVasyQuasilinearKdS}. One then has a partial expansion of the form \eqref{EqPfResExp}, with $\ell-\alpha<\Im\sigma_j<\ell$ and a remainder $\wt V\in e^{(\ell-\alpha)t_*}\CIb$, see the proof of \cite[Theorem~5.6]{HintzVasyQuasilinearKdS}. As long as $\ell-\alpha>0$, the above arguments show that the partial expansion must vanish identically, implying an improved decay rate for $V=\wt V$. Iterating this argument, one obtains $V\in e^{\ell t_*}\CIb$ for all $\ell>0$, and in the next iterative step, one obtains the partial expansion \eqref{EqPfKilling}, as desired.
\end{proof}

Rescaling $V$ and redefining $R$ and $\wt V$, we thus have
\[
  V=\pa_{t_*}+R+\wt V.
\]
Denote by $\phi_s$ the time $s$ flow of $V$, and let $\phi_s^0$ be the time $s$ flow of $\pa_{t_*}+R$, so $\phi_s^0(t_*,x)=(t_*+s,e^{s R}x)$.

\begin{lemma}
\label{LemmaPfLongTime}
  Define the neighborhoods
  \[
    \cU_j=(r_0+(8-j)\rho,r_1-(8-j)\rho)\times\Sph^2 \subset\R^3, \quad j=1,2
  \]
  of $\ol{\Sigma_{9\rho}}$. If the data $D$ are sufficiently close to $D_{b_0}$ in a sufficiently high $\cC^m$ norm, then:
  \begin{enumerate}
  \item \label{ItPfLongTimeExists} The flow $\phi_s$ maps $[0,\infty)_{t_*}\times\cU_1\to[0,\infty)_{t_*}\times\cU_2$ for all $s\geq 0$.
  \item \label{ItPfLongTimeConv} Let $\cW=(\frac{1}{2},\frac{3}{2})_{t_*}\times\cU_1$, and write
  \[
    \phi_s(t_*,x)=\phi_s^0(\psi_s(t_*,x))
  \]
  for $(t_*,x)\in\cW$, $s\geq 0$. Then the limit
  \[
    \psi_\infty(t_*,x):=\lim_{s\to\infty}\psi_s(t_*,x)\in\R^4
  \]
  exists and defines a smooth function satisfying $\|\psi_\infty-\Id\|_{L^\infty}<\rho$. Moreover, the convergence is exponentially fast, that is, $\psi_s(t_*,x)-\psi_\infty(t_*,x)\in e^{-\alpha s}\cCb^0([0,\infty)_s;\CI(\cW;\R^4))$.
  \end{enumerate}
\end{lemma}
\begin{proof}
  The estimates $|V t_*-1|\leq\eps e^{-\alpha t_*}<1/2$, $|V r|\leq\eps e^{-\alpha t_*}$, with $\eps>0$ small for $D$ close to $D_{b_0}$, imply that for the $r$-coordinate $r(\phi_s(t_*,x))$, we have
  \[
    |r(\phi_s(t_*,x))-r(\phi_0(t_*,x))| \leq \int_0^s \eps e^{-\alpha(s'/2)}\,ds' \leq \frac{2\eps}{\alpha}<\rho.
  \]
  This proves \itref{ItPfLongTimeExists}. Next, note that $\psi_s(t_*,x)=(\phi_s^0)^{-1}(\phi_s(t_*,x))$ is well-defined for all $s\geq 0$ by part \itref{ItPfLongTimeExists}. Estimates similar to the previous one show that, for $\eps$ small enough,
  \[
    \|\psi_s-\Id\|_{\cCb^0([0,\infty)_s;\cC^0(W;\R^4))}<\rho,
  \]
  and moreover
  \begin{equation}
  \label{EqPfLongTimePsiAPriori}
    \psi_s-\Id \in \cCb^0([0,\infty)_s;\cC^k(W;\R^4))
  \end{equation}
  for all $k\in\Z_{\geq 0}$. Decomposing $\wt V=\wt V^T\pa_{t_*} + \wt V^X$, with $\wt V^X$ a spatial vector field, and writing $\psi_s=(\psi_s^T,\psi_s^X)$ for the time and spatial components, respectively, $\psi_s$ satisfies the ODE
  \[
    \pa_s\psi_s = (\wt V^T,(e^{-s R})_*\wt V^X)(\psi_s^T+s,e^{s R}\psi_s^X)
  \]
  with initial condition $\psi_0=\Id$. In view of \eqref{EqPfLongTimePsiAPriori} and the decay properties of $\wt V$, this gives $\pa_s\psi_s\in e^{-\alpha s}\cCb^0([0,\infty)_s;\cC^k(W;\R^4))$ for all $k$, and hence the conclusion.
\end{proof}

By definition of $V$, we have $\phi_s^*(g',F')=(g',F')$ in $\cW$ for all $s$. Hence, taking the limit as $s\to\infty$ in the topology of $\CI(\cW)$, we have
\[
  (g',F') = \lim_{s\to\infty}\bigl(\phi_s^*(g_b,F_b) + \phi_s^*(\wt g,\wt F)\bigr) = \psi_\infty^*(g_b,F_b),
\]
where for the second equality we used that the pull-back by $\phi_s^0$ preserves $(g_b,F_b)$. The proof of Theorem~\ref{ThmIntroUniq} is complete.


\end{document}